\documentclass{article}
\usepackage{amssymb,amsmath,epsfig,eepic}
\def\openone{\leavevmode\hbox{\small1\kern-3.3pt\normalsize1}}
\def\im{\mbox{Im\,}}

\def\span{\mbox{span\,}}
\def\diag{\mbox{diag\,}}

\def\const{\mbox{const\,}}
\def\ad{\mbox{ad\,}}
\def\tr{\mbox{tr\,}}

\def\bbbz{\mathbb{Z}}

\def\im{\mbox{Im\;}}

\def\hgt{\mbox{ht\;}}
\def\bq{{\bf q}}

\newtheorem{theorem}{Theorem}
\newtheorem{lemma}{Lemma}
\newtheorem{proof}{Proof}
\newtheorem{remark}{Remark}
\title{On Kaup-Kupershchmidt type equations and their soliton solutions}
\author{ Vladimir S. Gerdjikov $^1$ \\
{\it $^{1}$ Institute for Nuclear Research and Nuclear Energy } \\
{\it Bulgarian Academy of Sciences,} \\ {\it72 Tsarigradsko chaussee, 1784 Sofia, BULGARIA }\\
 }

\begin{document}

\maketitle

\begin{abstract}
We start with the Lax representation for the Kaup-Kupersschmidt  equation (KKE). We
We outline the deep relation between the scalar Lax operator and the matrix Lax operators
related to Kac-Moody algebras. Then we derive the MKdV equations gauge equivalent to the KKE.
Next we outline the symmetry and the spectral properties of the relevant Lax operator.
Using the dressing Zakharov-Shabat method we demonstrate that the MKdV and KKE have
two types of one-soliton solutions and briefly comment on their properties.
\end{abstract}

\section{Introduction}

In 1980 D. Kaup and Satsuma \cite{Kaup,SatsKau} analyzed integrable NLEE related to
the scalar third-order operator:
\begin{equation}\label{eq:Ka1}
\mathcal{L} {\boldmath \psi}\equiv \frac{\partial ^3 {\boldmath \psi}}{ \partial x^3 } + 6Q(x,t) \frac{\partial {\boldmath \psi}}{ \partial x }
+ (6R(x,t) - \lambda^3){\boldmath \psi} =0,
\end{equation}
and demonstrated that it allows two interesting reductions each giving rise  to an interesting generaliztion of KdV \cite{CauDodGib}:
\begin{equation}\label{eq:Ka1re}
 \mbox{A)} \qquad R=0, \qquad  \mbox{B)} \qquad R= \frac{1}{2}\frac{\partial Q}{ \partial x }.
\end{equation}
In the first case the operator $\mathcal{L}_{\rm A}$ combined with
\begin{equation}\label{eq:MA1}
\mathcal{M}_{\rm A}  {\boldmath \psi} \equiv \frac{\partial {\boldmath \psi}}{ \partial t } - \left(9\lambda^3 - 18 \frac{\partial^2 Q}{ \partial x^2 }\right)
\frac{\partial ^2 {\boldmath \psi}}{ \partial x^2 } - 6 \left( \frac{\partial ^2 Q}{ \partial x^2 } - 6Q^2\right) \frac{\partial {\boldmath \psi}}{ \partial x }
- 36 \lambda^3 Q {\boldmath \psi}=0,
\end{equation}
leads to Sawada-Kotera equation \cite{SatsKau,Kaup,SawKot}:
\begin{equation}\label{eq:Ka2}
\frac{\partial Q}{ \partial t }  + \frac{\partial^5 Q }{ \partial x^5 } + 30 \left( \frac{\partial ^3 Q}{ \partial x^3 } Q + \frac{\partial^2 Q }{ \partial x^2 }
\frac{\partial Q}{ \partial x } \right) + 180 \frac{\partial Q}{ \partial x } Q^2 =0.
\end{equation}
Here and below we have replaced Kaup's spectral parameter $\lambda$ by $\lambda^3$.
In the second case the Lax operator $\mathcal{L}_{\rm B}$ with
\begin{equation}\label{eq:MB1}
\mathcal{M}_{\rm B}  {\boldmath \psi} \equiv \frac{\partial {\boldmath \psi}}{ \partial t } - 9\lambda^3 \frac{\partial ^2 {\boldmath \psi}}{ \partial x^2 }
+ 3 \left( \frac{\partial ^2 Q}{ \partial x^2 } + 12Q^2\right) \frac{\partial {\boldmath \psi}}{ \partial x } - 3\left( \frac{\partial^3 Q }{ \partial x^3 }
+ 12 \lambda^3 Q \lambda+ 24 \frac{\partial Q}{ \partial x } Q \right) {\boldmath \psi}=0,
\end{equation}
provide the Lax representation for
\begin{equation}\label{eq:KKe}
\frac{\partial Q}{ \partial t }  + \frac{\partial^5 Q }{ \partial x^5 } + 30 \left( \frac{\partial ^3 Q}{ \partial x^3 } Q + \frac{5}{2} \frac{\partial^2 Q }{ \partial x^2 }
\frac{\partial Q}{ \partial x } \right) + 180 \frac{\partial Q}{ \partial x } Q^2 =0.
\end{equation}
known today as the Kaup-Kuperschmidt equation \cite{Kaup,Kuper}.

It is easy to check that there is no elementary change $\{ {\boldmath \psi}, x,t \} \to \{\tilde{\boldmath \psi}, \tilde{x}, \tilde{t}\} $ which
could transform eq. (\ref{eq:Ka2}) into (\ref{eq:KKe}). Both these equations are inequivalent also to the KdV5 equation, which takes the
form:
\begin{equation}\label{eq:KdV5}
\frac{\partial Q}{ \partial t }  + \frac{\partial^5 Q }{ \partial x^5 } + 30 \left( \frac{\partial ^3 Q}{ \partial x^3 } Q + 2 \frac{\partial^2 Q }{ \partial x^2 }
\frac{\partial Q}{ \partial x } \right) + 270 \frac{\partial Q}{ \partial x } Q^2 =0.
\end{equation}

The answer to the question why these three equations, so similar, are inequivalent was soon discovered.
It can be traced back to two seminal papers. The first one is by Mikhailov \cite{Mik} who introduced the notion of
the reduction group and discovered the integrability of the 2-dimensional Toda field theories (TFT). In fact, the equations (\ref{eq:Ka2})
and (\ref{eq:KKe}) belong to the hierarchies of NLEE containing 2-dimensional TFT related to the algebra $sl(3)$ but
requiring different inequivalent reductions.

In the other important paper Drinfeld and Sokolov \cite{DriSok} demonstrated the deep relations between the scalar operators
of the form (\ref{eq:Ka1}) and the first order matrix ordinary differential operators with deep reductions:
\begin{equation}\label{eq:L1}
L \psi \equiv i \frac{\partial \psi}{ \partial x } + U(x,t,\lambda) \psi(x,t,\lambda)=0, \qquad U(x,t,\lambda)=(q(x,t) - \lambda J),
\end{equation}
related to the Kac-Moody algebras. In fact, imposing the condition that $U(x,t,\lambda) = q(x,t) -\lambda J$ belongs to a
certain Kac-Moody algebra is equivalent to imposing on $U(x,t,\lambda)$ the relevant reduction in the sense of Mikhailov.
However, even until now the KKE and Sawada-Kotera equations still attract attention, especially concerning the
construction and the properties of their solutions, see \cite{Fuchs,Bagder,Das,DyParker,ForGib,Oiler,Parker,Parker2,Popov}

The paper is organized as follows. In Section 2 we describe the gauge transformations that underly the deep
relation between the scalar Lax operator (\ref{eq:Ka1}) and its equivalent (\ref{eq:LL}) and (\ref{eq:L23}).
In Section 3 we construct the Lax representation for the MKdV equations gauge equivalent to KKE. In the next
Section 4 we introduce the fundamental analytic solutions (FAS) of $L$ (\ref{eq:LL}) and the relevant Riemann-Hilbert
problem which underlies the inverse scattering problem for $L$. In Section 5 we apply the dressing Zakharov-Shabat method
\cite{ZaSha} for deriving the soliton solutions of the MKdV. In Section 6 we demonstrate that the poles of the dressing factor
and its inverse are in fact discrete eigenvalues of the Lax operator $L$. We end with discussion and conclusions.

\section{Preliminaries}

\subsection{From scalar to matrix operators}

First we will demonstrate how one can relate to a scalar ordinary differential operator
 a first order matrix ordinary operator \cite{DriSok}. We will work this out on the example of the
operator $\mathcal{L}$ (\ref{eq:Ka1}). Indeed, let us consider the 3-component vector
$\tilde{\psi}(x,t,\lambda) = (\boldmath{\psi}_{xx} + 6Q\boldmath{\psi}, \boldmath{\psi}_x, \boldmath{\psi})^T $
and let us differentiate it with respect to $x$. Using equation (\ref{eq:Ka1}) one easily finds that the third order
scalar operator $\mathcal{L}$ is equivalent to:
\begin{equation}\label{eq:Ltil}\begin{aligned}
\tilde{L}_{(1)} \tilde{\psi}(x,t,\lambda) &= \frac{\partial \tilde{\psi}}{ \partial x }  + \tilde{U}_{(1)}(x,t,\lambda) \tilde{\psi}(x,t,\lambda) =0,
&\quad \tilde{U}_{(1)}(x,t,\lambda) &= \tilde{q}_{(1)}(x,t) - \mathcal{J}_{(1)}(\lambda), \\
\tilde{q}_{(1)}(x,t) &= \left(\begin{array}{ccc} 0 & 0 & 6 (R - Q_x) \\ 0 & 0 & 6Q \\ 0 & 0 & 0 \end{array}\right), &\quad
\mathcal{J}_{(1)}(\lambda) &= \left(\begin{array}{ccc} 0 & 0 & \lambda^3 \\ 1 & 0 & 0 \\ 0 & 1 & 0  \end{array}\right).
\end{aligned}\end{equation}

Similarly any ordinary scalar differential operator of order $n$ can be rewritten as a first order $n\times n$ operator of
special type: all $x$-dependent coefficients occupy just the last column of $\tilde{q}(x,t)$ \cite{DriSok}.

\subsection{The relation between the matrix operators and Kac-Moody algebras}

The next important step was proposed by Drinfeld and Sokolov \cite{DriSok}. They proposed to apply to $\tilde{L}$ a
gauge transformation \cite{ZaMi}
so that the new operator $L$ acquires canonical form from the point of view of Kac-Moody algebras.

This we will do in two steps. The first step  takes $\tilde{L}_{(1)}$ to the operator $\tilde{L}_{(2)}$ with potential
\begin{equation}\label{eq:gtr}
\tilde{U}_{(2)}(x,t,\lambda) = g^{-1} \tilde{U}_{(1)}(x,t,\lambda) g + g^{-1} g_x,
\end{equation}
where $g(x,t)$ is a $3\times 3$ upper-triangular matrix of the form:
\begin{equation}\label{eq:g}
 g(x,t) = \left(\begin{array}{ccc} 1 & \bq_1 & c \\ 0 & 1 & \bq_2 \\ 0 & 0 & 1  \end{array}\right).
\end{equation}
The constraints that we impose on $\bq_1$, $\bq_2$ and $c$ require that
\begin{equation}\label{eq:U2}
 \tilde{U}_{(2)}(x,t,\lambda) = q_{(2)}(x,t) - \mathcal{J}_{(1)}(\lambda),
\end{equation}
where $q_{(2)}(x,t)$ is a diagonal matrix. This can be achieved imposing:
\begin{equation}\label{eq:g2}\begin{aligned}
c(x) &= \bq_1\bq_2 - \bq_1^2 - \frac{\partial \bq_1}{ \partial x }, \qquad 6Q(x) =
 -\frac{\partial \bq_1 }{ \partial x } -  \frac{\partial \bq_2}{ \partial x } -  \bq_1^2 - \bq_2^2 + \bq_1 \bq_2.,   \\
6R(x) &=\bq_1\bq_2 (\bq_1-\bq_2) +\frac{\partial \bq_1}{ \partial x } \bq_2 -2\bq_2 \frac{\partial \bq_2}{ \partial x }- \frac{\partial^2 \bq_2}{ \partial x^2 }
\end{aligned}\end{equation}
With this choice for $Q$ and $R$ we find $q_{(2)}(x,t) =\diag (\bq_1, \bq_2-\bq_1,-\bq_2)$. Note that $\tr (q(x,t))=0$,
i.e. it belongs to the Cartan subalgebra of the Lie algebra $sl(3)$.

The second step is to apply to $\tilde{L}_{(2)} $ a similarity transformation by the diagonal matrix
$C_2(\lambda) = \diag( \lambda^{-1}, 1 , \lambda)$. Thus we obtain the operator
\begin{equation}\label{eq:LL}\begin{aligned}
\tilde{L} &= C_2(\lambda) \tilde{L}_{(2)} C_2^{-1}(\lambda) \equiv \frac{\partial }{ \partial x } + \tilde{U}(x,t,\lambda), &\quad
\tilde{U}(x,t,\lambda) &= \tilde{ q}(x,t) - \lambda \tilde{ J}, \\
\tilde{ q}(x,t) &= \left(\begin{array}{ccc} \bq_1 & 0 & 0 \\ 0 & \bq_2-\bq_1 & 0 \\ 0 & 0 & -\bq_2 \end{array}\right), &\quad
\tilde{ J}&= \left(\begin{array}{ccc} 0 & 0 & 1 \\ 1 & 0 & 0 \\ 0 & 1 & 0 \end{array}\right).
\end{aligned}\end{equation}
Similar transformations can be applied also to the $M$-operator in the Lax pair.

\subsection{Relation to Kac-Moody algebras}

Let us now explain the relation of the above operator $\tilde{ L}$ (\ref{eq:LL}) to the Kac-Moody algebras \cite{Kac}.
Skipping the details we will outline the construction of the Kac-Moody algebras. Here we will assume that the
reader is familiar with the theory of simple Lie algebras \cite{Helg}. In fact it will be enough to know the Cartan-Weyl basis
of the algebra $sl(3)$, in Cartan classification this algebra is denoted as $A_2$.

The algebra $sl(3)$ has rank 2 and its root system contains three positive roots: $\Delta_+=\{ e_1-e_2, e_2-e_3, e_1-e_3\} $
and three negative roots  $\Delta_-=\{ -e_1+e_2,- e_2+e_3, -e_1+e_3\} $. The first two positive roots $\alpha_1 = e_1-e_2$ and
$\alpha_2= e_2-e_3$ form the set of simple roots of $sl(3)$. The third positive root is $\alpha_{\rm max} =e_1-e_3$ is the maximal one.
We will say that the simple roots are of height 1, the maximal root $\alpha_{\rm max} = \alpha_1+\alpha_2$ is of height 2. Similarly, the negative roots
$-\alpha_1$ and $-\alpha_2$ have height $-1$ and the minimal root $\alpha_{\rm min} =-\alpha_{\rm max}$ has height $-2 = 1 \mod(3)$.

The Cartan-Weyl basis of $sl(3)$ is formed by the Cartan subalgebra $\mathfrak{h}$ and by the Weyl generators $E_\alpha$ and $E_{-\alpha}$,
$\alpha\in \Delta_+$. In the typical $3\times 3$ representation these are given by:
\begin{equation}\label{eq:CWb}\begin{aligned}
 H_{\alpha_1} &=  \left(\begin{array}{ccc} 1 & 0 & 0 \\ 0 & -1 & 0 \\ 0 & 0 & 0   \end{array}\right), &\;
  H_{\alpha_2} &=  \left(\begin{array}{ccc} 0 & 0 & 0 \\ 0 & 1 & 0 \\ 0 & 0 & -1   \end{array}\right), &\;
E_{\alpha_1} &=  \left(\begin{array}{ccc} 0 & 1 & 0 \\ 0 & 0 & 0 \\ 0 & 0 & 0   \end{array}\right), \\
E_{\alpha_2} &=  \left(\begin{array}{ccc} 0 & 0 & 0 \\ 0 & 0 & 1 \\ 0 & 0 & 0   \end{array}\right), &\;
E_{\alpha_{\rm max}} &=  \left(\begin{array}{ccc} 0 & 0 & 1 \\ 0 & 0 & 0 \\ 0 & 0 & 0   \end{array}\right), &\;
E_{-\alpha} &= E_\alpha^T.
\end{aligned}\end{equation}
The main tool in constructing the Kac-Moody algebra based on $sl(3)$ is the grading, which according to \cite{Kac} must be
performed with the Coxeter automorphism. In our case the grading consists in splitting the algebra $sl(3)$ into the direct sum
of three linear subspaces as follows:
\begin{equation}\label{eq:grad}\begin{aligned}
sl(3) &\simeq  \mathfrak{g}^{(0)} \oplus \mathfrak{g}^{(1)} \oplus \mathfrak{g}^{(2)} , &\; \mathfrak{g}^{(0)} &\equiv \span \{H_{\alpha_1}, H_{\alpha_2} \},\\
 \mathfrak{g}^{(1)} &\equiv \span \{E_{\alpha_1}, E_{\alpha_2},
 E_{\alpha_{\rm min}} \}, &\; \mathfrak{g}^{(2)} &\equiv \span \{E_{-\alpha_1}, E_{-\alpha_2},  E_{\alpha_{\rm max}}\}.
\end{aligned}\end{equation}
In other words if we assume that the Cartan generators have height 0, then each of the subspaces $\mathfrak{g}^{(k)}$ consists
of elements of height $k$ modulo 3, which is the Coxeter number of $sl(3)$.  The important property of the subspaces $\mathfrak{g}^{(k)}$
is provided by the relation: if $X_k \in \mathfrak{g}^{(k \mod 3)}$ and $X_m \in \mathfrak{g}^{(m \mod 3)}$ then
\begin{equation}\label{eq:KMal}\begin{aligned}
{} [X_k, X_m ] \in  \mathfrak{g}^{(k+m \mod 3)}.
\end{aligned}\end{equation}

 Next the elements of the Kac-Moody algebra based on this grading of $sl(3)$ consists of finite or semi-infinite series of the form
 \begin{equation}\label{eq:KMel}\begin{aligned}
 X(\lambda) = \sum_{p\ll N}^{} X_p \lambda^p, \qquad X_p \in \mathfrak{g}^{(p \mod 3)}.
 \end{aligned}\end{equation}
The subspaces $\mathfrak{g}^{(p)}$ are in fact the eigensubspaces of the Coxeter automorphism $\tilde{ C}_0$ which in this case
is an element of the Cartan subgroup of the form
\begin{equation}\label{eq:Cox}\begin{aligned}
 \tilde{ C}_0= \exp \left( \frac{2\pi i}{3} H_\rho \right), \qquad \rho = e_1 -e_3; \qquad \mbox{i.e.} \qquad \tilde{ C}_0=\diag (\omega, 1, \omega^{-1}).
\end{aligned}\end{equation}
Indeed, it is easy to check that $\tilde{ C}_0 E_\alpha \tilde{ C}_0^{-1} = \omega^k E_\alpha$, where $\omega =e^{2\pi i/3}$ and
$k = \hgt (\alpha)$. Obviously $\tilde{ C}_0^3 =\openone$.

\begin{remark}\label{rem:2}
In fact we will use also an alternative grading, used also in \cite{DriSok} in which the subspaces $\mathfrak{g}^{(1)}$ and $\mathfrak{g}^{(2)}$ are
interchanged. It is generated by an equivalent realization of the Coxeter automorphism: $\tilde{ C}_0^{-1} E_\alpha \tilde{ C}_0 = \omega^{-k} E_\alpha$,
where $k=\hgt (\alpha)$.
\end{remark}
\begin{remark}\label{rem:1}
The gauge transformation described above is the analogue of the famous Miura transformation, which maps
the KdV equation into the modified KdV (MKdV) equation. Therefore the Lax pair (\ref{eq:LM}) will
produce not the KKE, but rather a system of MKdV eqs. that are gauge equivalent to KKE.

\end{remark}

In order to understand the interrelation between  the Kac-Moody algebras and the ordinary differential operators it remains to
note that the potential $\tilde{ U}(x,t,\lambda)$ in (\ref{eq:LL}) is an element of the Kac-Moody algebra $A_2^{(1)}$ with the the alternative grading,
see Remark~\ref{rem:2}.

\subsection{Factorized ordinary differential operators}

In fact the matrix operator $\tilde{ L}$ (\ref{eq:LL}) can be brought back to scalar form. Indeed, the matrix scattering problem $\tilde{ L}\chi(x,t,\lambda)=0$
can be written down as
\begin{equation}\label{eq:Lchi}\begin{aligned}
\frac{\partial \chi_1}{ \partial x } + q_1(x,t) \chi_1 (x,t,\lambda) &= \lambda \chi_3 (x,t,\lambda), \\
\frac{\partial \chi_2}{ \partial x } + (q_2(x,t) - q_1(x,t)) \chi_2 (x,t,\lambda) &= \lambda \chi_1 (x,t,\lambda), \\
\frac{\partial \chi_3}{ \partial x } - q_2(x,t) \chi_3 (x,t,\lambda) &= \lambda \chi_2 (x,t,\lambda),
\end{aligned}\end{equation}
which can easily be rewritten as the following scalar eigenvalue problem:
\begin{equation}\label{eq:Lfac}\begin{aligned}
\tilde{\mathcal{L}} \chi_3  &= \lambda^3 \chi_3 (x,t,\lambda), \\
\tilde{\mathcal{L}} &= (D_x +q_1(x,t)) (D_x +(q_1(x,t)-q_2(x,t)) (D_x -q_2(x,t)) ,
\end{aligned}\end{equation}
which is  a factorized third order differential operator $\tilde{\mathcal{L}}$.
Thus the gauge transformation (\ref{eq:gtr}) effectively takes the scalar operator $\mathcal{L}$ (\ref{eq:Ka1}) into the
factorized one $\tilde{\mathcal{L}}$ (\ref{eq:Lfac}).

\section{Lax representation of MKdV eqs}

\subsection{The generic two-component MKdV's equations}
Now we have the tools to construct the Lax pair of the MKdV equation which is gauge equivalent to KKE.
In this Section we will derive the Lax representation of the fifth order MKdV equations. This derivation, as well as the
direct and inverse scattering problems for these Lax operators are more conveniently executed if the term $\lambda \tilde{J}$
is taken in diagonal form. This is easily achieved by a similarity transformation with the constant matrix $w_0$:
\begin{equation}\label{eq:L23}\begin{aligned}
L = w_0^{-1} \tilde{L}w_0, \qquad w_0 = \frac{1}{\sqrt{3}}\left(\begin{array}{ccc} \omega & 1 & \omega^2 \\ 1 & 1 & 1 \\
\omega^2 & 1 & \omega \end{array}\right)
\end{aligned}\end{equation}

Thus we construct the Lax pair for the MKdV eqs. as follows:
\begin{equation}\label{eq:LM}\begin{aligned}
L \psi (x,t,\lambda) &\equiv i\frac{\partial \psi }{ \partial x } + \left( q(x,t) -\lambda J\right)\psi(x,t,\lambda)=0, \\
M \psi (x,t,\lambda) &\equiv i\frac{\partial \psi }{ \partial t } + \left( V(x,t,\lambda) -\lambda^5 K\right)\psi(x,t,\lambda)=0, \\
V(x,t,\lambda) &= \sum_{s=0}^{4} V_{s}(x,t) \lambda^s, \qquad J= \diag (\omega, 1,\omega^2),  \qquad K= \diag (\omega^2, 1,\omega),
\end{aligned}\end{equation}
where the basis in the algebra $sl(3)$ is given in the Appendix. Below  we use the notations:
\begin{equation}\label{eq:B10.2}\begin{aligned}
q(x,t) & = i\sqrt{3} (\bq_1 B_1^{(0)} + \bq_2 B_2^{(0)} ), &\qquad J &= \omega^2 B_3^{(1)}, \qquad K=  a \omega B_3^{(2)}, \\
q_1 &= \omega \bq_1 +\omega^{-1}\bq_2, &\qquad  q_2 &= -(\omega^{-1} \bq_1 +\omega \bq_2),
\end{aligned}\end{equation}
and $\omega= \exp (2\pi i /3)$.

\subsection{Solving the recurrent relations and $\Lambda$-operators}

The condition $[L,M]=0$ must hold true identically with respect to $\lambda$. This leads to a set of equations:
\begin{equation}\label{eq:rec1}\begin{aligned}
& \lambda^5 &\qquad [K,q] &= [J,V_4] , \\
& \lambda^k &\qquad  i\frac{\partial V_k}{ \partial x } + [q, V_k] &= [J,V_{k-1}] , \qquad k=1,\dots, 4; \\
& \lambda^0 &\qquad  i\frac{\partial V_0}{ \partial x } - i \frac{\partial q}{ \partial t } &=0.
\end{aligned}\end{equation}
The equations (\ref{eq:rec1}) can be viewed as recurrent relations which allow one to determine $V_k$ in terms of $q$
and its $x$-derivatives. This kind of problems have been thoroughly analyzed, see \cite{GYa*94,Yan1}, so below we just
briefly mention the effects of the $\mathbb{Z}_3$ reductions we have imposed.

Obviously we need to split each od the coefficients $V_k \in \mathfrak{g}^{(\underline{k})}$ into diagonal and offdiagonal parts:
\begin{equation}\label{eq:Vk2}\begin{aligned}
V_k = V_k ^{\rm f} + w_k B_3 ^{(\underline{k})}, \qquad \underline{k} = k \mod 3.
\end{aligned}\end{equation}
From the appendix it is clear that only $B_3^{(1)}$ and  $B_3^{(2)}$ are non-vanishing, while, for example,  $V_0$ and  $V_3$
do not have diagonal parts, so  $V_0 \equiv V_0^{\rm f}$ and $V_3 \equiv V_3^{\rm f}$.

The linear mapping $\ad_J \cdot \equiv [J, \cdot ]$ is also playing an important role. The diagonal and the off-diagonal parts of the
matrices in fact provide the kernel and the image of $\ad_J$. So in the space of off-diagonal matrices we can define also the
inverse $\ad_J^{-1}$. In addition $\ad_J$ and $\ad_J^{-1}$ map the linear spaces $\mathfrak{g}^{(\underline{k})}$ as follows:
\begin{equation}\label{eq:adJ}\begin{aligned}
\ad_J \colon   \mathfrak{g}^{(\underline{k})} \to  \mathfrak{g}^{(\underline{k+1})}, \qquad
\ad_J^{-1} \colon   \mathfrak{g}^{(\underline{k})} \to  \mathfrak{g}^{(\underline{k-1})}.
\end{aligned}\end{equation}

The formal solution of the recurrent relations (\ref{eq:rec1}) is most conveniently written down with the help of the
recursion operators:
\begin{equation}\label{eq:recop}\begin{aligned}
\Lambda_0 Y &= \ad_J^{-1} \left( i \frac{\partial Y}{ \partial x } + [q, Y]^{\rm f} \right), \qquad Y \in \mathfrak{g}^{(\underline{0})}, \\
\Lambda_k X_k^{\rm f} &= \ad_J^{-1} \left( i \frac{\partial X_k^{\rm f}}{ \partial x } + [q, X_k^{\rm f}]^{\rm f} + \frac{i}{3} [q, B_3^{(\underline{k})}]
\int_{ }^{x} dy\; \left\langle [q(y), X_k^{\rm f}], B_3^{(\underline{3-k})}\right\rangle \right),
\end{aligned}\end{equation}
where $X_k  \in \mathfrak{g}^{(\underline{k})} $. Thus the formal solution of the recurrent relations provides the
following answer for $V_k$:
\begin{equation}\label{eq:Vk-form}\begin{aligned}
V_4 &= \ad_J^{-1} [K,q], & \qquad V_3^{\rm f} &=  \Lambda_0 V_4, \\
V_2^{\rm f} & = \Lambda_2 V_3^{\rm f} + \frac{i}{3} w_2 B_3^{(2)}, &\quad
V_1^{\rm f} & = \Lambda_1 V_2^{\rm f} + \frac{i}{3} w_1 B_3^{(1)}, &\quad
V_0^{\rm f} &= \Lambda_0 V_1^{\rm f},
\end{aligned}\end{equation}
where
\begin{equation}\label{eq:wk}\begin{aligned}
 w_k = \int_{}^{x} dy \; \left\langle [q(y), V_k^{\rm f}], B_3^{(\underline{3-k})}\right\rangle .
\end{aligned}\end{equation}
At the end we get also the formal expression for the corresponding NLEE:
\begin{equation}\label{eq:nle}\begin{aligned}
i \frac{\partial q}{ \partial t }  + a \frac{\partial }{ \partial x } {\bf \Lambda} \Lambda_0 \ad_J^{-1} [K,q]=0, \qquad {\bf \Lambda} =
\Lambda_0 \Lambda_1 \Lambda_2.
\end{aligned}\end{equation}
Obviously, one can consider as a potential to the  $M$-operator polynomial $V_{(N)} = \sum_{p=0}^{N}V_p \lambda^p$
of any power $N$ as long as $N+1 \neq 0 \mod 3$. The corresponding NLEE will be generated by a relevant polynomial
of the recursion operators $\Lambda_k$.

\subsection{The explicit form of the $M$-operator}
Let us now do the calculations for the $V_k$ explicitly. Skipping the details we have:
\begin{equation}\label{eq:V12.3}\begin{aligned}
 V_4(x,t) &= -ia\sqrt{3}  (q_1 \omega^2  B_1^{(1)} +q_2    B_2^{(1)} ),
\end{aligned}\end{equation}
\begin{equation}\label{eq:V12.4}\begin{aligned}
 V_3(x,t) &= a  (v_{3,1}  B_1^{(0)}  - v_{3,2}  B_2^{(0)}, \qquad
 v_{3,1}  = \frac{\partial q_1}{ \partial x } +3 q_2^2, \qquad  v_{3,2}  = \frac{\partial q_2}{ \partial x } -3 q_1^2,
\end{aligned}\end{equation}
\begin{equation}\label{eq:V12.5}\begin{aligned}
 V_2(x,t) &= \frac{ia\sqrt{3}}{9}  (v_{2,1} \omega  B_1^{(2)} + v_{2,2} B_2^{(2)} + v_{2,3}   B_3^{(2)} ), \\
  v_{2,1} &= \frac{\partial^2 q_1 }{ \partial x^2 } +6 q_2 \frac{\partial q_2}{ \partial x }, \qquad
  v_{2,2} = \frac{\partial^2 q_2 }{ \partial x^2 } -6  q_1 \frac{\partial q_1}{ \partial x }, \\
  v_{2,3} &= 3 q_1 \frac{\partial q_2 }{ \partial x } - 3 q_2 \frac{\partial q_1}{ \partial x } + 6(q_1^3 +q_2^3).
\end{aligned}\end{equation}

\begin{equation}\label{eq:V12.6}\begin{aligned}
 V_1(x,t) &= \frac{a}{3} ( -v_{1,1} \omega^2  B_1^{(1)} + v_{1,2}  B_2^{(1)} +v_{1,3} \omega^2  B_3^{(1)}), \\
   v_{1,1} &=  \frac{\partial^3 q_1 }{ \partial x^3 } + 3 q_2 \frac{\partial^2 q_2}{ \partial x^2 } + 6 \left( \frac{\partial q_2}{ \partial x }\right)^2
  + 27 q_1q_2 \frac{\partial q_1}{ \partial x } -9 q_1^2 \frac{\partial q_2}{ \partial x } 18 q_1 (q_1^3 + q_2^3),   \\
    v_{1,2} &= 3 \frac{\partial^3 q_2 }{ \partial x^3 } - 3 q_1 \frac{\partial^2 q_1}{ \partial x^2 } - 6 \left( \frac{\partial q_1}{ \partial x }\right)^2
  + 27 q_1q_2 \frac{\partial q_2}{ \partial x } -9 q_2^2 \frac{\partial q_1}{ \partial x } -18 q_2 (q_1^3 + q_2^3),   \\
  v_{1,3}   &=  3 \left( q_2\frac{\partial^2 q_1 }{ \partial x^2 } +  q_1 \frac{\partial^2 q_2}{ \partial x^2 } -   \frac{\partial q_1}{ \partial x }\frac{\partial q_2}{ \partial x }
  -3  q_1^2 \frac{\partial q_1}{ \partial x } +3 q_2^2 \frac{\partial q_2}{ \partial x } + 9 q_1^2 q_2^2 \right).
\end{aligned}\end{equation}

\begin{equation}\label{eq:V12.7}\begin{aligned}
 V_0(x,t) &= -\frac{i a\sqrt{3}}{9 }  (v_{0,1} \omega  B_1^{(0)} +v_{0,2}  B_2^{(0)}), \\
  v_{0,1} &=  \frac{\partial^4 q_1 }{ \partial x^4 } + 15 \frac{\partial q_2}{ \partial x } \frac{\partial^2 q_2}{ \partial x^2 }
  + 45 q_2 \left( \left( \frac{\partial q_1}{ \partial x }\right)^2 + q_1 \frac{\partial ^2 q_1}{ \partial x^2 } \right)  \\
&\qquad\qquad  +  45 (q_1^3 +q_2^3) \frac{\partial q_1}{ \partial x } +27  q_2^2 (5q_1^3 + 2q_2^3), \\
  v_{0,2} &=  \frac{\partial^4 q_2 }{ \partial x^4 } - 15 \frac{\partial q_1}{ \partial x } \frac{\partial^2 q_1}{ \partial x^2 }
  + 45 q_1 \left( \left( \frac{\partial q_2}{ \partial x }\right)^2 + q_2 \frac{\partial ^2 q_2}{ \partial x^2 } \right) \\
  &\qquad \qquad -  45 (q_1^3 +q_2^3) \frac{\partial q_2}{ \partial x } + 27 q_1^2 (2q_1^3 + 5q_2^3),
\end{aligned}\end{equation}

The NLEE:
\begin{equation}\label{eq:Nle}\begin{aligned}
 \frac{\partial q_1}{ \partial t} + \frac{a}{9} \frac{\partial v_{01}}{ \partial x }=0, \qquad
  \frac{\partial q_2}{ \partial t} + \frac{a}{9} \frac{\partial v_{02}}{ \partial x }=0.
\end{aligned}\end{equation}

\subsection{Special reductions}

Kaup considers two special reductions on his La operator $\mathcal{L}$: A) $R=0$ and B) $R= \frac{1}{2}Q_x$.
It is the second reduction that is responsible for the KKE; in terms $q_1$ and $q_2$ it can be formulated as
\begin{equation}\label{eq:q1-2}\begin{aligned}
 q_2 =-q_1.
\end{aligned}\end{equation}
It may be realized using external automorphism of $sl(3)$, so it must be responsible for $A_2^{(2)}$ Kac-Moody algebra.

Imposing the reduction (\ref{eq:q1-2})  we get the equation:
\begin{equation}\label{eq:eq-q1}\begin{aligned}
\frac{\partial q_1}{ \partial t } &= -\frac{a}{9} \frac{\partial }{ \partial x } \left(
 \frac{\partial^4 q_1 }{ \partial x^4 } + 15 \frac{\partial q_1}{ \partial x } \frac{\partial^2 q_1}{ \partial x^2 }
  - 45 q_1 \left( \left( \frac{\partial q_1}{ \partial x }\right)^2 + q_1 \frac{\partial ^2 q_1}{ \partial x^2 } \right)  + 81q_1^5 \right),
\end{aligned}\end{equation}
which is gauge equivalent to the KKE.

\section{The FAS of the Lax operators with $\mathbb{Z}_3 $-reduction.}\label{sec:2.3}

The idea for the FAS for the generalized Zakharov-Shabat (GZS) system has been proposed by Shabat \cite{Sha}, see also \cite{ZMNP}.
However for the GZS $J$ is with real eigenvalues, while our Lax operator has complex eigenvalues.
The ideas of Shabat were generalized by Beals and Coifman \cite{CBC2} and Caudrey \cite{CBC1} for operators $L$
related to the  algebras $sl(n)$; these results were  extended to $L$ operators related to any simple Lie algebra $\mathfrak{g}$
\cite{GYa*94}, see also \cite{ConM,Yan1,SIAM,GeYa-13}.

The Jost solutions of eq. (\ref{eq:LM}) are defined by:
\begin{equation}\label{eq:La3}\begin{aligned}
\lim_{x \to -\infty}\phi_+(x ,\lambda)e^{i \lambda Jx } & = \openone, \qquad
\lim_{x \to \infty}\phi_-(x ,\lambda)e^{i \lambda Jx }  = \openone,
\end{aligned}\end{equation}
They satisfy the integral equations:
\begin{equation}\label{eq:Jo1}\begin{aligned}
Y_\pm (x ,\lambda) = \openone + \int_{\pm \infty}^{x } dy e^{-i\lambda J(x -y)}  Q(y) Y_\pm (y,\lambda)  e^{i\lambda J(x -y)} ,
\end{aligned}\end{equation}
where $Y_\pm (x ,\lambda) = \phi_\pm (x ,\lambda)e^{i\lambda Jx }$. Unfortunately, with our choice for $J=\diag (\omega,1,\omega^2)$
this integral equations have no solutions. The reason is that the factors $e^{i\lambda J(x -y)} $ in the kernel in (\ref{eq:Jo1})
can not be made to decrease simultaneously.

Following the ideas of Caudrey, Beals and Coifman, see \cite{CBC1,CBC2,GYa*94} we start with the Jost solutions for potentials on compact
support, i.e. assume that $q(x )=0$ for $x  < -L_0$ and $x  >L_0$. Then the integrals in (\ref{eq:Jo1}) converge and one can prove
the existence of $Y_\pm (x ,\lambda)$.

The continuous spectrum of $L$ consists of those points $\lambda$, for which $e^{i\lambda (J_k - J_j)(x-y)}$ oscillate, which means that
\begin{equation}\label{eq:CS1}
\im \lambda (J_k - J_j) = \im \lambda (\omega^{2-k} - \omega^{2-j}) =0.
\end{equation}
It is easy to check that for each pair of indices $k \neq j$ eq. (\ref{eq:CS1}) has a solution of the form $\arg \lambda =\const $ depending on
$k$ and $j$.
The solutions for all choices of the pairs $k,j$ fill up a pair of rays $l_\nu$ and $l_{\nu+3}$ which are given by:
\begin{eqnarray}\label{eq:l-nu}
l_\nu \colon \arg(\lambda) ={\pi (2\nu +1) \over 6 } , \qquad \Omega_\nu \colon
\frac{\pi (2\nu+1) }{6} \leq \arg \lambda \leq \frac{\pi (2\nu+3) }{6} ,
\end{eqnarray}
where $\nu =0,\dots , 5$, see Fig. \ref{fig:1}.

Thus the  analyticity regions of the FAS are the 6 sectors  $\Omega_\nu$, $\nu=0,\dots, 5$ split up by the set of rays
$l_\nu$, $\nu =0,\dots ,5$, see Fig. \ref{fig:1}. Now we will outline how one can construct a FAS in each of these sectors.

Obviously, if $\im \lambda \alpha(J)=0$ on the rays $l_\nu\cup l_{\nu+3}$, then $\im \lambda \alpha(J)>0$ for
$\lambda \in \Omega_\nu \cup \Omega_{\nu+1} \cup \Omega_{\nu+2}$ and $\im \lambda \alpha(J)<0$ for
$\lambda \in \Omega_{\nu-1} \cup \Omega_{\nu-2} \cup \Omega_{\nu-3}$; of course all indices here are understood modulo 6. As a result the
factors  $e^{-i\lambda J(x -y)} $ will decay exponentially if $\im \alpha(J) <0$ and $x  -y >0$ or if $\im \alpha(J) >0$ and $x  -y <0$.
In eq. (\ref{eq:10}) below we have listed the signs of $\im \alpha(J) $ for each of the sectors $\Omega_\nu$.

\begin{figure}
  \centering
  \includegraphics[width=6.5cm]{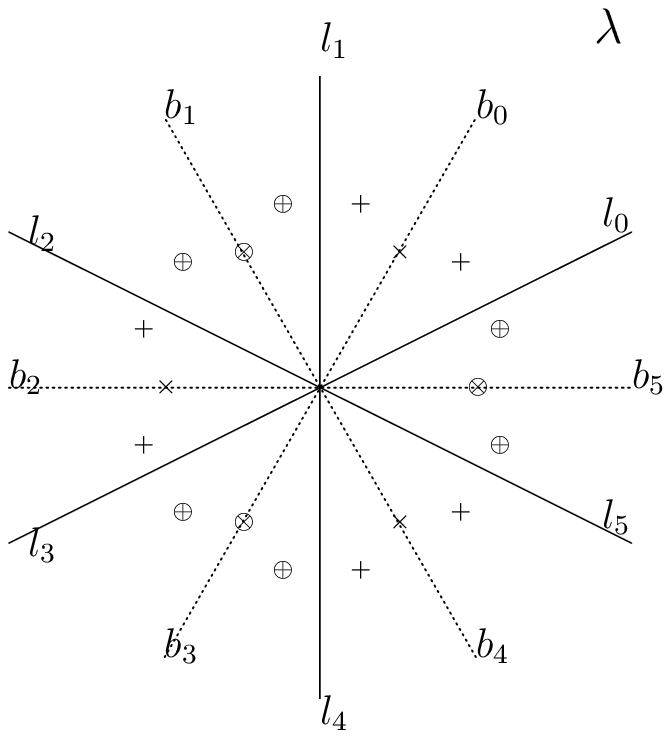}
  \caption{The contour of the RHP with $\bbbz_3$-symmetry fills up the rays $l_\nu$, $\nu=1,\dots, 6$.
By $\times$ and $\otimes$ (resp. by $+$ and $\oplus$)  we have denoted  the locations of the
discrete eigenvalues corresponding to a soliton of first type (resp. of second type).
\label{fig:1}}

\end{figure}

To each ray one can relate the root satisfying $\im \lambda \alpha(J)=0 $, i.e.
\begin{equation}\label{eq:10}\begin{aligned}
& l_0,  &\quad \pm &(e_1-e_2) &\quad \Omega_0 &\quad &\alpha_1<0, &\quad &\alpha_2 >0 &\quad &\alpha_3 >0 \\
& l_1,  &\quad \pm &(e_1-e_3) &\quad \Omega_1 &\quad &\alpha_1 >0, &\quad &\alpha_2 >0 &\quad &\alpha_3 <0 \\
& l_2,  &\quad \pm &(e_2-e_3) &\quad \Omega_2 &\quad &\alpha_1 <0, &\quad &\alpha_2 <0 &\quad &\alpha_3 <0 .
\end{aligned}\end{equation}
There are two fundamental regions: $\Omega_0$ and $\Omega_1$.
The FAS in the other sectors can be obtained from the FAS in $\Omega_0$ and $\Omega_1$  by acting with the automorphism $C_0$:
\begin{equation}\label{eq:11}\begin{aligned}
C_0\Omega_\nu &\equiv \Omega_{\nu+2}, &\qquad C_0 l_\nu &\equiv l_{\nu+2}, &\quad   \nu =0,1, \dots , 5.
\end{aligned}\end{equation}

The next step is to construct the set of integral equations for FAS
which will be analytic in $\Omega _\nu  $. They are different from the integral
equations for the Jost solutions  (\ref{eq:Jo1}) because for each choice of the
matrix element $(k,j)$ we specify the lower limit of the integral so that all
exponential factors $e^{i\lambda (J_k-J_j)(x  -y)}$ decrease for $x ,y\to \pm \infty$,
\begin{equation}\label{eq:12.1nu}\begin{aligned}
X^\nu _{kj}(x ,\lambda ) &= \delta _{kj} + i \int_{\epsilon_{kj}\infty }^{x} dy
e^{-i\lambda (J_k-J_j)(x  -y)} \sum_{p=1}^{h} q_{kp}(y) X^\nu _{pj}(y,\lambda ) ,
\end{aligned}\end{equation}
where the signs $\epsilon_{kj}$ for each of the sectors $\Omega_\nu$ are collected
in the table \ref{tab:1}, see also \cite{Tih,ConM,GeYa-13}. We also assume that for $k=j$ $\epsilon_{kk}=-1$.

\begin{table}
\centering
\begin{tabular}{|c|c|c|c|c|c|c|}
  \hline
  $(k,j)$ & (1,2)  & (1,3) & (2,3) & (2,1) & (3,2) & (3,1) \\ \hline
  $\Omega_0$ & $-$ & $+$ & $+$ & $+$ & $-$ & $-$ \\
  $\Omega_1$ &  $-$ & $+$ & $-$ & $+$ & $-$ & $+$ \\
  $\Omega_2$ & $-$ & $+$ & $-$ & $+$ & $-$ & $+$ \\
  $\Omega_3$ & $+$ & $+$ & $+$ & $-$ & $-$ & $-$ \\
  $\Omega_4$ & $-$ & $+$ & $-$ & $+$ & $-$ & $+$ \\
  $\Omega_5$ & $-$ & $+$ & $-$ & $+$ & $+$ & $-$ \\
  \hline
\end{tabular}
\caption{ The set of signs $\epsilon_{kj}$ for each of the sectors $\Omega_\nu$. \label{tab:1}}
\end{table}

The solution of the integral equations (\ref{eq:12.1nu})
will be the FAS of  $L$ in the sector $\Omega _\nu  $. The asymptotics
of $X^\nu (x,\lambda ) $ and $X^{\nu -1} (x,\lambda )  $ along the
ray $l_\nu  $ can be written in the form \cite{GYa*94,GeYa-13}:
\begin{equation}\label{eq:xi-as}\begin{aligned}
\lim_{x\to -\infty } e^{i\lambda J x} X^\nu (x,\lambda e^{i0} ) e^{-i\lambda J x} &= S_\nu ^+ (\lambda ), &\quad \lambda \in l_\nu  ,\\
\lim_{x\to \infty } e^{i\lambda J x} X^\nu (x,\lambda e^{i0} ) e^{-i\lambda J x} &= T_\nu ^- (\lambda )D_\nu ^+(\lambda ), &\quad \lambda \in l_\nu  ,\\
\lim_{x\to -\infty } e^{i\lambda J x} X^{\nu -1} (x,\lambda e^{-i0} )  e^{-i\lambda J x} &= S_\nu ^- (\lambda ), &\quad \lambda \in l_\nu  , \\
\lim_{x\to \infty } e^{i\lambda J x} X^{\nu -1} (x,\lambda e^{-i0} )  e^{-i\lambda J x} &= T_\nu ^+ (\lambda )D_\nu ^-(\lambda ), &\quad \lambda \in l_\nu  ,
\end{aligned}\end{equation}
where the matrices $S_\nu ^\pm $ and $T_\nu ^\pm $ belong to $su(2)$ subgroups of $sl(3)$. More specifically
from the integral equations (\ref{eq:12.1nu}) we find:
\begin{equation}\label{eq:SpmTpm0}\begin{aligned}
S_0^+ (\lambda) &= \openone + s^+_{0;21}E_{21}, &\; T_0^- (\lambda) &= \openone + \tau^-_{0;12} E_{12}, \\
S_0^- (\lambda) &= \openone +  s^+_{0;12}E_{12}, &\; T_0^+ (\lambda) &= \openone +   \tau^+_{0;21}E_{21}, \\
D_0^+ (\lambda) &=  d^+_{0;1} E_{11} + \frac{1}{d^+_{0;1}}E_{22}+E_{33}, &\; D_0^- (\lambda) &= \frac{1}{d^-_{0;1}} E_{11} +d^-_{0;1} E_{22}+E_{33}.
\end{aligned}\end{equation}
and
\begin{equation}\label{eq:SpmTpm1}\begin{aligned}
S_1^+ (\lambda) &= \openone +   s^+_{1;31} E_{31}, &\; T_1^- (\lambda) &= \openone + \tau^-_{1;13} E_{13}, \\
S_1^- (\lambda) &= \openone +  s^+_{1;13}E_{13}, &\; T_1^+ (\lambda) &= \openone + \tau^+_{1;31} E_{31}, \\
D_1^+ (\lambda) &=  d^+_{1;1} E_{11} +E_{22}+ \frac{1}{d^+_{1;1}} E_{33}, &\; D_1^- (\lambda) &= \frac{1}{d^-_{1;1}} E_{11}+ E_{22} +d^-_{1;1} E_{33}.
\end{aligned}\end{equation}
By $E_{kj}$ we mean a $3\times 3$ matrix with matrix elements $(E_{kj})_{mn}=\delta_{um} \delta_{jn}$.

The corresponding factors for the asymptotics of $X^\nu (x,\lambda e^{i0} )$ for $\nu>1$ are obtained from eqs. (\ref{eq:SpmTpm0}),
 (\ref{eq:SpmTpm1}) by applying the automorphism $C_0$. If we consider potential on finite support, then we can define not only
 the Jost solutions $\Psi_\pm (x,\lambda)$ but also the scattering matrix $T(\lambda) = \phi_-(x,\lambda)\phi_+^{-1}(x,\lambda)$.
The  factors $S_\nu^\pm (\lambda)$, $T_\nu^\pm (\lambda)$ and $D_\nu^\pm (\lambda)$ provide an analog of the Gauss decomposition
of the scattering matrix with respect to the $\nu  $-ordering, i.e.:
\begin{equation}\label{eq:nu-gauss}
T_\nu (\lambda ) = T_\nu ^-(\lambda ) D_\nu ^+(\lambda ) \hat{S}_\nu ^+
(\lambda ) =  T_\nu ^+(\lambda ) D_\nu ^-(\lambda ) \hat{S}_\nu ^-
(\lambda ) , \qquad  \lambda \in l_\nu .
\end{equation}

The $\mathbb{Z}_n $-symmetry imposes the following constraints on the FAS and
on the scattering matrix and its factors:
\begin{equation}\label{eq:Z_n-cons}\begin{aligned}
 C_0 X^\nu (x,\lambda \omega ) C_0^{-1} &= X^{\nu -2} (x,\lambda ) ,
&\qquad C_0 T_\nu (\lambda \omega ) C_0^{-1} &= T_{\nu -2}(\lambda ), \\
C_0 S^\pm_\nu (\lambda \omega ) C_0^{-1} &= S^\pm_{\nu -2}(\lambda ),
&\qquad C_0 D^\pm_\nu (\lambda \omega ) C_0^{-1} &= D^\pm_{\nu -2}(\lambda) ,
\end{aligned}\end{equation}
where the index $\nu -2 $ should be taken modulo $6 $.
Consequently we can view as independent only the data on two of the
rays, e.g. on $l_0  $ and $l_{1}$; all the rest will be
recovered using the reduction conditions.

If in addition we impose the $\mathbb{Z}_2 $-symmetry, then we will have also:
\begin{equation}\label{eq:Z_2-cons}\begin{aligned}
& \mbox{a)} &\; K_0^{-1}(X^\nu (x,-\lambda ^*))^\dag K_0 &= \hat{X}^{N+1-\nu}(x,\lambda), &\;
K_0^{-1}(S_\nu ^\pm (-\lambda ^*))K_0 &= \hat{S}_{N+1-\nu}^{\mp}(\lambda ),\\[3pt]
& \mbox{b)} &\; K_0^{-1}(X^\nu (x,\lambda ^*))^* K_0 &= \hat{X}^{\nu}(x,\lambda),
 &\; K_0^{-1}(S_\nu ^\pm (\lambda ^*))K_0 &= \hat{S}_{N+1-\nu}^{\mp}(\lambda ),
\end{aligned}\end{equation}
where $K_0 = E_{1,3} + E_{2,2} + E_{3,1}$ and by `hat' we denote the inverse matrix.
Analogous relations hold true for $T_\nu ^\pm(\lambda )$ and $D_\nu ^\pm(\lambda ) $.
One can prove also that $D_\nu ^+(\lambda ) $ (resp. $D_\nu
^-(\lambda ) $) allows analytic extension for $\lambda \in \Omega
_\nu  $ (resp. for $\lambda \in \Omega _{\nu -1} $. Another important fact is
that  $D_\nu ^+(\lambda ) = D_{\nu +1}^-(\lambda ) $ for all
$\lambda \in \Omega _\nu  $.

The next important step is the possibility to reduce the solution of the
ISP for the GZSs to a (local) RHP. More precisely, we have:
\begin{equation}\label{eq:*-nu}\begin{aligned}
X^\nu (x,t ,\lambda ) &= X^{\nu -1}(x,t ,\lambda ) G_\nu (x,t ,\lambda ), &\; \lambda &\in l_\nu ,\\
G_\nu (x,t ,\lambda ) &= e^{i (\lambda Jx +\lambda^5 Kt)} G_{0,\nu }(\lambda ) e^{-i (\lambda Jx +\lambda^5 Kt)},  &\; G_{0,\nu
}(\lambda ) &= \left. \hat{S}_\nu ^- S_\nu ^+(\lambda )\right|_{t=0} .
\end{aligned}\end{equation}
The collection of all these relations for $\nu =0,1,\dots,5 $ together with
\begin{equation}\label{eq:*-nu-norm}
\lim_{\lambda \to\infty } X^\nu (x,t ,\lambda ) = \openone ,
\end{equation}
can be viewed as a local RHP posed on the collection of rays
$\Sigma \equiv \{l_\nu \}_{\nu =1}^{2N} $ with canonical normalization.
Rather straightforwardly we can  prove that if $X^\nu (x,\lambda ) $ is a solution of the
RHP  then $\chi ^\nu (x,\lambda)=X^\nu (x,\lambda ) e^{-i\lambda J x } $ is a FAS of $L$ with potential
\begin{equation}\label{eq:q-CBC}
q(x ,t) = \lim_{\lambda \to\infty } \lambda \left( J - X^\nu (x ,t ,\lambda ) J \hat{X}^\nu (x ,t ,\lambda ) \right).
\end{equation}

\section{The dressing method and the $N$-soliton solutions}
The main idea of the dressing method \cite{ZaSha,Mik,Mik} is, starting from a known regular solution of the RHP $X_0^\nu(x,t,\lambda)$ to
construct a new singular solution $X_1^\nu(x,t,\lambda)$ of the same RHP. The two solutions are related by a dressing factor
$u(x,t,\lambda)$
\begin{equation}\label{eq:u}\begin{aligned}
X_1^\nu(x,t,\lambda) = u(x,t,\lambda) X_0^\nu(x,t,\lambda),
\end{aligned}\end{equation}
which may have pole singularities in $\lambda$. A typical anzats for $u(x,t,\lambda)$ is given by \cite{Mik}, see also \cite{115a,121a,GMSV}:
\begin{equation}\label{eq:6NuNp''m}\begin{aligned}
 u(x  ,t  ,\lambda ) &=\openone + \sum_{s=0}^{2} \left(
 \sum_{l=1}^{N_1} \frac{ C^{-s}A_{l} C^s }{\lambda -\lambda_l \omega^s}+\sum_{r=N_1+1}^{N}
 \frac{ C^{-s} A_{r} C^s}{\lambda -\lambda_r \omega^s}+\sum_{r=N_1+1}^{N} \frac{C^{-s} A^*_{r}C^s}{\lambda -(\lambda_r^*) \omega^s} \right)
\end{aligned}\end{equation}
with  $3N_1 +6N_2$ poles
and $\lambda_p$ is real if $p\in\overline{1,N_1}$ and complex if $p\in\overline{N_1+1,N_1+N_2}$.
In \cite{115a} both types of simplest one-soliton solutions for the Tzitzeica equation are derived. Note that
Tzitzeica equation possesses Lax representation with the same Lax operator, but its $M$-operator is linear with respect to $\lambda^{-1}$.

The dressing factor $u(x,t,\lambda)$ satisfies the equation:
\begin{equation}\label{eq:ux}\begin{aligned}
i \frac{\partial u}{ \partial x }  + (q^{(1)}(x,t) -\lambda J) u(x,t,\lambda) - u(x,t,\lambda)  (q^{(0)}(x,t) -\lambda J) =0,
\end{aligned}\end{equation}
where $q^{(0)}(x,t)$ (typically chosen to be vanishing) corresponds to the `naked' Lax operator (or to the regular solution of RHP),
while $q^{(1)}(x,t)$ is the potential of the `dressed' Lax operator. The equation (\ref{eq:ux}) must hold true identically
with respect to $\lambda$. We also assume that the residues $A_k(x ,t )$ are degenerate matrices of the form:
\begin{equation} \begin{aligned}
 A_k(x ,t ) = |n_k(x ,t )\rangle \langle m_k^T(x ,t )|, \qquad
(A_k)_{ij}(x ,t ) = n_{k;i}(x ,t )m_{k;j}(x ,t ).
\end{aligned}\end{equation}
Thus $u(x ,t ,\lambda)$ for $N_1=N_2=1$ has $9$  poles located at $\lambda_1 \omega^k$ with $\lambda_1$ real and $\lambda_2 \omega^k$,
$\lambda^*_2 \omega^k$, with $k=0,1,2$ and $\lambda_2$ complex.

Evaluating the residue of eq. (\ref{eq:ux}) for $\lambda =\lambda_k$ one finds that the `polarization' vectors $|n_k(x ,t )\rangle$ and
$ \langle m_k^T(x ,t )|$ must satisfy the equations:
\begin{equation}\label{eq:n-km-k}\begin{aligned}
 i \frac{\partial |n_k\rangle}{ \partial x }  + (\tilde{q}^{(1)}(x,t) - \lambda_k \tilde{ J})  |n_k(x ,t )\rangle =0, \qquad
i \frac{\partial  \langle m_k^T|}{ \partial x } + \lambda_k   \langle m_k^T(x ,t )| \tilde{J} =0.
\end{aligned}\end{equation}
where we have put $q^{(0)}(x,t)=0 $. Then the vectors  $\langle m_k^T(x ,t )|$ must depend on $x$ and $t$ as folows:
\begin{equation}\label{eq:m2}\begin{aligned}
m_1&=\mu_2  e^{-2\mathcal{X}_1}+2 |\mu_1| e^{\mathcal{X}_1 }\cos \left( \Omega_1 -\frac{2\pi}{3} \right),  &\;
m_2&= \mu_2  e^{-2\mathcal{X}_1}+2 |\mu_1| e^{\mathcal{X}_1 }\cos \left( \Omega_1  \right)   \\
m_3&=\mu_2  e^{-2\mathcal{X}_1}+2 |\mu_1| e^{\mathcal{X}_1 }\cos \left( \Omega_1 +\frac{2\pi}{3} \right), &\;
 \mathcal{X}_1&=\frac{1}{2}(\lambda_1 x + \lambda_1^5 t), \\
 \Omega_1 &=\frac{\sqrt{3}}{2}(\lambda_1 x - \lambda_1^5 t)-\alpha_1 , &\;   \mu_1&=|\mu_1|e^{i\alpha_1}.
\end{aligned}\end{equation}

The factors $u(x,t, \lambda)$ must also satisfy all symmetry conditions characteristic for the FAS. The $\mathbb{Z}_3$ symmetry
is already taken into account with the anzatz (\ref{eq:6NuNp''m}).
From  the second $\bbbz_2$-reduction (\ref{eq:Z_2-cons}), $K_0^{-1} u^\dag (x , t , -\lambda^* )K_0 = u^{-1} (x ,t , \lambda )$,
after taking the limit $\lambda \to\lambda_k$, we obtain  algebraic equation for $|n_k\rangle$ in terms of $\langle m_k^T|$:
Below we list the relevant formulae just  for the two types of one-soliton solutions (for the general case see \cite{Mik,115a}):
\begin{equation}\begin{aligned}\label{eq:4.3}
\mbox{a)} \qquad  |n_1\rangle = A^{-1}|m_1\rangle , \qquad  \mbox{b)} \qquad \left(\begin{array}{c} |n_2\rangle \\ |n_2^*\rangle
 \end{array}\right)  = \left(\begin{array}{cc} D & F \\ F^* & D^*   \end{array}\right)^{-1} \left(\begin{array}{c} |m_2\rangle \\ |m_2^*\rangle
 \end{array}\right) ,
\end{aligned}\end{equation}
where case a)  corresponds to the choice $N_1=1$, $N_2=0$  while case b) is relevant for  $N_1=0$, $N_2=1$.
The notation above are as follows:
\begin{equation}\label{eq:6N2s3mm}\begin{aligned}
A  &=\frac {1}{2\lambda_1^3 } \diag (Q^{(1)} , Q^{(2)} , Q^{(3)} ), &\quad
D  &=\frac {1}{2\lambda_2^3 } \diag (P^{(1)}, P^{(2)}, P^{(3)}), \\
F  &=\frac {1}{\lambda_2^3 + \lambda_2^{*,3}} \diag (K^{(1)}, K^{(2)}, K^{(3)}), &\;
Q^{(j)}  & = \langle m_1^T| \Lambda_{11}^{(j)}(\lambda_l, \lambda_1) |m_1 \rangle , \\
K^{(j)}  & = \langle m_2^{*,T}| \Lambda_{12}^{(j)}(\lambda_1, \lambda^*_2) |m_1 \rangle , &\;
P^{(j)}  & = \langle m_2^T| \Lambda_{21}^{(j)}(\lambda_2, \lambda_1) |m_l \rangle,
\end{aligned}\end{equation}
with
\begin{equation}\label{eq:QKP}\begin{aligned}
\Lambda^{(j)}_{lp} & =-\lambda_l \lambda_p E_{1+j,3-j}+\lambda_l^2 E_{2+j,2-j}+\lambda_p^2 E_{3+j, 1-j},  \quad j=1,2,3.
\end{aligned}\end{equation}

Skipping  the details, we just mention that this approach allows one to obtain the explicit form of the $N$-soliton solutions. We just mention
that along with the explicit expressions for the vectors $|n_k\rangle$ in terms of $\langle m_j| $ that follow from eqs. (\ref{eq:4.3})--(\ref{eq:QKP})
and  take into account that $| m_j\rangle $ are solutions of the `naked' Lax operator with vanishing potential $q^{(0)} =0$.

We end this section with a few comments about the simplest one-soliton solutions of the MKdV and KKE equations.
The first one is  these one-soliton solutions with generic choice of the polarization vectors {\em are not} traveling waves.

Skipping the details (see e.g. \cite{115a}) we find for the naked solution of the Lax operator
\begin{equation}\label{eq:h2}\begin{aligned}
\bq_1(x)=-\partial_x \ln \left( \frac{n_3m_3}{\lambda_1}-1\right)= -\partial_x \ln \left( \frac{2m_1 m_3}{m_2^2}-1\right)
\end{aligned}\end{equation}

Note that the soliton solution of the  KKE can be obtained from (\ref{eq:h2}) by:
\begin{equation}\label{eq:6Q}\begin{aligned}
 6Q=-2\bq_{1,x}-\bq_1^2.
\end{aligned}\end{equation}
 In the special case $\mu_2 =0$ we have
\begin{equation}\label{eq:link2}\begin{aligned}
 \bq_1(x,t)=-\partial_x \ln \left( \frac{1}{2}+ \frac{3}{2}\tan ^2 \Omega_1\right)= -\frac{3\sqrt{3} \lambda_1 \tan \Omega_1 }{2\sin ^2 \Omega_1 + 1 }
\end{aligned}\end{equation}

This solution is obviously singular. In addition the relevant potential of the Lax operator
\begin{equation}\label{eq:Q1}\begin{aligned}
6Q=\frac{9\lambda_1^2 (1-4\sin^2 \Omega_1)}{(1+2\sin^2 \Omega_1)^2}
\end{aligned}\end{equation}
is not in the functional class, since it does not decays  to zero.

A possible way to find regular soliton solutions of these equations is to take into account the
fact that both MKdV and KKE are invariant under the transformation
\[ x \to i x', \qquad t\to i t' ,  \qquad Q \to -Q'. \]
Then $ \Omega_1 \to i \Omega_1'$ and we get:
\begin{equation}\label{eq:link3}\begin{aligned}
\bq_1(x,t)= \frac{3\sqrt{3} \lambda_1 \tanh \Omega_1' }{2\sinh ^2 \Omega_1' -1  }, \qquad 6Q'=-\frac{9\lambda_1^2 (1+4\sinh^2 \Omega_1')}{(1-2\sinh^2 \Omega_1')^2}
\end{aligned}\end{equation}
which this time is singular only at two points $\sinh  \Omega_1'=\pm \frac{1}{\sqrt{2}}$.
 Moreover the potential $Q'(x,t)$  decays for  $x \to \pm \infty$

\section{The  Resolvent of the Lax operator}

The FAS can be used to construct the kernel of the resolvent of the Lax operator $L$.
In this section by $\chi^\nu (x ,\lambda)$ we will denote:
\begin{equation}\label{eq:1}\begin{aligned}
\chi^\nu (x ,\lambda) = u(x ,\lambda) \chi_0^\nu (x ,\lambda),
\end{aligned}\end{equation}
where $\chi_0^\nu (x ,\lambda)$ is a regular FAS and $u(x ,\lambda)$ is a dressing factor of
general form (\ref{eq:6NuNp''m}).

\begin{remark}\label{rem:3}
The dressing factor $u(x ,\lambda)$ has $3N_1+6N_2$ simple poles located at
$\lambda_l \omega^p$, $\lambda_r \omega^p$ and $\lambda_r^* \omega^p$ where $l=1,\dots, N_1$, $r=1,\dots , N_2$ and $p=0,1,2$.
Its inverse $u^{-1}(x ,\lambda)$ has also $3N_1+6N_2$ poles located $-\lambda_l \omega^p$, $-\lambda_r \omega^p$ and $-\lambda_r^* \omega^p$.
In what follows for brevity we will denote them by $\lambda_j $, $-\lambda_j$ for $j=1,\dots , 3N_1+6N_2$.
\end{remark}

Let us introduce
\begin{equation}\label{eq:R-nu}\begin{aligned}
R^\nu (x , x ',\lambda) &= \frac{1}{i} \chi^\nu (x ,\lambda) \Theta_\nu (x  -x ') \hat{\chi}^\nu (x ',\lambda),
\end{aligned}\end{equation}
\begin{equation}\label{eq:Theta}\begin{aligned}
\Theta_\nu (x  - x ') = \diag \left( t _\nu^{(1)} \theta( t _\nu^{(1)} (x  -x ')),
t _\nu^{(2)} \theta( t _\nu^{(2)} (x  -x ')), t _\nu^{(3)} \theta( t _\nu^{(3)} (x  -x ')) \right),
\end{aligned}\end{equation}
where  $\theta(x  -x ')$ is the step-function and $t _\nu^{(k)}=\pm 1$, see the table \ref{tab:2}.
\begin{table}
\centering
\begin{tabular}{|c|c|c|c|c|c|c|}
  \hline
  $ $ & $\Upsilon_0$  &  $\Upsilon_1$ &  $\Upsilon_2$ &  $\Upsilon_3$ &  $\Upsilon_4$ &  $\Upsilon_5$ \\ \hline
  $t _{\nu}^{(1)} $ & $-$ & $-$ & $-$ & $+$ & $+$ & $+$ \\
  $t _{\nu}^{(2)}$ &  $+$ & $+$ & $-$ & $-$ & $-$ & $+$ \\
  $t _{\nu}^{(3)}$ & $-$ & $+$ & $+$ & $+$ & $-$ & $-$ \\
  \hline
\end{tabular}
\caption{The set of signs $t _{\nu}^{(k)}$ for each of the sectors $\Upsilon_\nu$ (\ref{eq:Upsi}). \label{tab:2}}
\end{table}

\begin{theorem}\label{thm:1}
Let $Q(x ) $ be a Schwartz-type function  and let $ \lambda _j^\pm $
be the simple zeroes of the dressing factor $u(x , \lambda ) $ (\ref{eq:6NuNp''m}). Then

\begin{enumerate}

\item The functions $R^\nu (x , x ',\lambda)$ are analytic for $\lambda\in\Upsilon_\nu$ where
\begin{equation}\label{eq:Upsi}\begin{aligned}
b_\nu \colon \arg \lambda = \frac{\pi (\nu+1)}{3} , \qquad \Upsilon_\nu \colon
\frac{\pi (\nu+1)}{3} \leq \arg \lambda \leq \frac{\pi (\nu+2)}{3} .
\end{aligned}\end{equation}
 having pole singularities at $\pm \lambda _j^\pm $;

\item $R^\nu (x , x ',\lambda ) $ is a kernel of a bounded integral operator
for $\lambda \in \Upsilon_\nu $;

\item $R^\nu (x , x ',\lambda ) $ is uniformly bounded function for $\lambda
\in b_\nu $ and provides a kernel of an unbounded integral operator;

\item $R^\nu (x , x ',\lambda ) $ satisfy the equation:
\begin{equation}\label{eq:R3.1}
L(\lambda ) R^\nu (x , x ',\lambda )=\openone \delta (x - x ').
\end{equation}
\end{enumerate}
\end{theorem}

\begin{proof}[Idea of the proof] {}

\begin{enumerate}

\item First we shall prove that $R^\nu (x , x ',\lambda )$ has no jumps
on the rays $l_\nu$. From Section 3 we know that $X^\nu (x ,\lambda)$ and
therefore also $\chi^\nu (x ,\lambda)$ are analytic for $\lambda\in\Omega_\nu$.
So we have to show that the limits of $R^\nu (x , x ',\lambda )$ for $\lambda\to l_\nu$
from $\Upsilon_\nu$ and $\Upsilon_{\nu-1}$ are equal. Let show that for $\nu=0$. From
the asymptotics (\ref{eq:xi-as}) and from the RHP (\ref{eq:*-nu}) we have:
\begin{equation}\label{eq:2}\begin{aligned}
\chi^0 (x ,\lambda) = \chi^1 (x ,\lambda) G_1(\lambda), \qquad G_1 (\lambda)= \hat{S}_1^+(\lambda) S_1^-(\lambda), \quad \lambda\in l_1,
\end{aligned}\end{equation}
where $G_1(\lambda) $ belongs to an $SL(2)$ subgroup of $SL(3)$ and is such that it commutes with $\Theta_1(x  -x ')$.
Thus we conclude that
\begin{equation}\label{eq:3}\begin{aligned}
R_1(x ,x ' ,\lambda e^{+i0}) = R_1(x ,x ' ,\lambda e^{-i0}), \qquad \lambda\in l_1 .
\end{aligned}\end{equation}
Analogously we prove that $R_\nu(x ,x ' ,\lambda e^{+i0})$ has no jumps on the other rays $l_\nu$.

The jumps on the rays $b_\nu$ appear because of two reasons: first, because of the functions $\Theta_\nu (x  -x ')$
and second, it is easy to check that for $\lambda\in b_\nu$ the kernel $R_\nu(x ,x ' ,\lambda )$
 oscillates for $x ,x '$ tending to $\pm \infty$.
Thus on these lines the resolvent is unbounded integral operator.

\item Assume that $\lambda \in \Upsilon_\nu $ and consider the asymptotic behavior
of $R^\nu (x , x ',\lambda ) $ for $x , x '\to\infty  $. From equations
(\ref{eq:xi-as}) we find that
\begin{eqnarray}\label{eq:R3.2}
R_{ij}^\nu (x , x ',\lambda ) &=& \sum_{p=1}^{n} X^\nu_{ip}(x ,\lambda )
e^{-i\lambda J_p(x - x ')} \Theta_{\nu;pp}(x  -x ') \hat{X}^{\nu}_{pj}(x ',\lambda )
\end{eqnarray}

Due to the fact that $\chi_\nu (x ,\lambda ) $ has the special triangular asymptotics
for $x \to\infty  $ and $\lambda \in \Upsilon_\nu $ and for the correct choice of
$\Theta_\nu(x  -x ') $ (\ref{eq:Theta}) we check that the right hand side of
(\ref{eq:R3.2}) falls off exponentially for $x \to\infty  $ and arbitrary
choice of $x ' $. All other possibilities are treated analogously.

\item For $\lambda \in b_\nu$ the arguments of 2) can not be applied
because the exponentials in the right hand side of (\ref{eq:R3.2})
$\im \lambda =0  $ only oscillate. Thus we conclude that
$R^\nu(x , x ',\lambda ) $ for $\lambda \in b_\nu $ is only a bounded
function and thus the corresponding operator $R(\lambda ) $ is an
unbounded integral operator.

\item The proof of eq. (\ref{eq:R3.1}) follows from the fact that
$L(\lambda )\chi_\nu(x ,\lambda )=0 $ and
\begin{equation}\label{eq:R4.1}
\frac{\partial \Theta(x  -x ')}{ \partial x  }= \openone \delta (x  -x ').
\end{equation}
\end{enumerate}
\end{proof}

\begin{lemma}\label{lem:1}
The poles of $R^\nu (x , x ',\lambda)$ coincide with
the poles of the dressing factors $u(x ,\lambda)$ and its inverse $u^{-1}(x ,\lambda)$.
\end{lemma}

\begin{proof}
The proof follows immediately from the definition of $R^\nu (x , x ',\lambda)$ and from Remark \ref{rem:3}.
\end{proof}

Thus we have established that dressing by the factor $u(x ,\lambda)$, we in fact add to the discrete spectrum
of the Lax operator $6N_1+12N_2$ discrete eigenvalues; for $N_1=N_2=1$ they are shown on Figure \ref{fig:1}.

\section{Discussion and Conclusions}
On the example of the KKE we analyzed the relation between the scalar ordinary differential operators and the
Kac-Moody algebras. Using the dressing method we established that KKE and its gauge equivalent MKdV have two types
of one-soliton solutions, which generically {\em are not} traveling wave solutions. The dressing method adds discrete
eigenvalues to the spectrum of $L$ which comes in sextuplets for each soliton of first type and in dodecaplets (12-plets) for the
solitons of second type. Still open is the question of constructing {\em regular} soliton solutions and to study
the pro perties of the generic one soliton solutions that {\em are. not} traveling waves.

We have constructed the FAS of $L$ which satisfy a RHP on the set of rays $l_\nu$.
We also constructed the resolvent of the Lax operator and proved that its continuous spectrum
fills up the rays $b_\nu$ rather than $l_\nu$. From Figure \ref{fig:1} we see that the eigenvalues
corresponding to the  solitons of first type lay on the continuous spectrum of $L$. This explains
why the solitons of first type are singular functions.

Using the explicit form of the resolvent $R^\nu (x , x ',\lambda)$ and the contour integration method one can  derive
the completeness relation of the FAS.

As a further development we note, that one can  use the expansions over the squared solutions \cite{Tih}
to derive the action-angle variables of the NLEE in the hierarchy. These expansions are, in fact,
spectral decompositions of the relevant recursion operators $\Lambda_k$ which in addition possess
important geometrical properties \cite{Yan1,YanVi,YanVi2}.

\section*{Acknowledgments}
The author is grateful to Dr. Alexander Stefanov and to Prof. R. Ivanov for useful discussions
and help in preparing the manuscript.

\appendix
\section{The basis of $A_2^{(1)}$}

The basis of $A_2^{(1)}$ is obtained from the Cartan-Weyl basis of $A_2$ (\ref{eq:CWb}) by taking the average
with the Coxeter automorphism. In this case the  Coxeter automorphism is represented by
$\tilde{C}X \tilde{C}^{-1}$,
\[ \tilde{C}  = \left(\begin{array}{ccc} 0 & 0 & 1 \\ 1 & 0 & 0\\ 0 & 1 & 0
 \end{array}\right) .\]

\begin{equation}\label{eq:Basis}\begin{aligned}
B_1^{(0)}&= \left(\begin{array}{ccc} 0 & 1 & 0 \\ 0 & 0 & 1\\ 1 & 0 & 0  \end{array}\right), &\;
B_2^{(0)}&= \left(\begin{array}{ccc} 0 & 0 & 1 \\ 1 & 0 & 0 \\ 0 & 1 & 0  \end{array}\right), \\
B_1^{(1)}&= \left(\begin{array}{ccc} 0 & 1 & 0 \\ 0 & 0 & \omega^2 \\ \omega & 0 & 0  \end{array}\right), &\;
B_2^{(1)}&= \left(\begin{array}{ccc} 0 & 0 & 1 \\ \omega^2 & 0 & 0\\ 0 & \omega & 0  \end{array}\right),  &\;
B_3^{(1)}&= \left(\begin{array}{ccc} \omega^2 & 0 & 0 \\ 0 & \omega & 0\\ 0 & 0 & 1  \end{array}\right), \\
B_1^{(2)}&= \left(\begin{array}{ccc} 0 & 1 & 0 \\ 0 & 0 & \omega \\ \omega^2 & 0 & 0  \end{array}\right), &\;  B_2^{(2)}&=
\left(\begin{array}{ccc} 0 & 0 & 1 \\ \omega & 0 & 0\\ 0 & \omega^2 & 0  \end{array}\right),  &\;  B_3^{(2)}&=
\left(\begin{array}{ccc} \omega & 0 & 0 \\ 0 & \omega^2 & 0\\ 0 & 0 & 1  \end{array}\right).
\end{aligned}\end{equation}

\end{document}